\documentclass[12pt,a4paper]{iopart}
\usepackage{graphicx}
\usepackage{amsthm}
\usepackage{iopams}
\usepackage{color}
\usepackage{amssymb}
\usepackage{bbm}
\usepackage{braket}

\newcommand{\ie}{\textit{i.e.}}
\newcommand{\eg}{\textit{e.g.}}
\newcommand\T{\rule{0pt}{3.2ex}}
\newcommand\B{\rule[-1.4ex]{0pt}{0pt}}

\newtheorem{proposition}{Proposition}[section]

\begin{document}

\title{Permutationally invariant state reconstruction}

\author{Tobias Moroder$^{1,2}$, Philipp Hyllus$^3$, G{\'e}za T{\'o}th$^{3,4,5}$, Christian Schwemmer$^{6,7}$, Alexander Niggebaum$^{6,7}$, Stefanie Gaile$^8$, Otfried G\"uhne$^{1,2}$ and Harald Weinfurter$^{5,6}$}   

\address{$^1$ Naturwissenschaftlich-Technische Fakult\"at, Universit{\"a}t Siegen, Walter-Flex-Stra{\ss}e~3, D-57068 Siegen, Germany}
\address{$^2$ Institut f\"ur Quantenoptik und Quanteninformation, \"Osterreichische Akademie der Wissenschaften, Technikerstra{\ss}e~21A, A-6020~Innsbruck, Austria} 
\address{$^3$ Department of Theoretical Physics, University of the Basque Country UPV/EHU, P.O.~Box~644, E-48080 Bilbao, Spain}
\address{$^4$ \textsc{Ikerbasque}, Basque Foundation for Science, E-48011 Bilbao, Spain}
\address{$^5$ Wigner Research Centre for Physics, Hungarian Academy of Sciences, P.O. Box 49, H-1525 Budapest, Hungary}
\address{$^6$ Max-Planck-Institut f\"ur Quantenoptik, Hans-Kopfermann-Stra{\ss}e 1, D-85748 Garching, Germany}
\address{$^7$ Fakult\"at f\"ur Physik, Ludwig-Maximilians-Universit\"at, D-80797 M\"unchen, Germany}
\address{$^8$ Technical University of Denmark, Department of Mathematics, Matematiktorvet Building 303 B, 2800 Kgs. Lyngby, Denmark}
\eads{\mailto{moroder@physik.uni-siegen.de}}
\pacs{03.65.Wj, 03.65.-w, 03.67.-a}

\begin{abstract}
Feasible tomography schemes for large particle numbers must possess, besides an appropriate data acquisition protocol, also an efficient way to reconstruct the density operator from the observed finite data set. Since state reconstruction typically requires the solution of a non-linear large-scale optimization problem, this is a major challenge in the design of scalable tomography schemes. Here we present an efficient state reconstruction scheme for permutationally invariant quantum state tomography. It works for all common state-of-the-art reconstruction principles, including, in particular, maximum likelihood and least squares methods, which are the preferred choices in today's experiments. This high efficiency is achieved by greatly reducing the dimensionality of the problem employing a particular representation of permutationally invariant states known from spin coupling combined with convex optimization, which has clear advantages regarding speed, control and accuracy in comparison to commonly employed numerical routines. First prototype implementations easily allow reconstruction of a state of 20 qubits in a few minutes on a standard computer.
\end{abstract}

\section{Introduction}

Full information about the experimental state of a quantum system is naturally highly desirable because it enables one to determine the mean value of each observable and thus also of each other property of the quantum state. Abstractly, such a complete description is given for example by the density operator, a positive semidefinite matrix~$\rho$ with unit trace. Quantum state tomography~\cite{qse_book} refers to the task to determine the density operator for a previously unknown quantum state by means of appropriate measurements. Via the respective outcomes more and more information about the true state generating the data is collected up to the point where this information uniquely specifies the particular state. Quantum state tomography has successfully been applied in many experiments using different physical systems, including trapped ions~\cite{haeffner05a} or photons~\cite{kiesel07a} as prominent instances. 

Unfortunately, tomography comes with a very high price due to the exponential scaling of the number of parameters required to describe composed quantum systems. For an $N$ qubit system the total number of parameters of the associated density operator is $4^N-1$ and any standard tomography protocol is naturally designed to determine all these variables. The most common scheme used in experiments~\cite{james01a} consists of locally measuring in the basis of all different Pauli operators and requires an overall amount of $3^N$ different measurement settings with $2^N$ distinct outcomes each. Other schemes, \eg, using an informationally complete measurement~\cite{renes04a} locally would require just one setting but, nevertheless, the statistics for $4^N$ different outcomes. Hence the important figure of merit to compare different methods is given by the combination of settings and outcomes. 

For such a scaling, the methods rapidly become intractable, already for present state-of-the-art experiments: Recording for example the data of $14$ trapped ions~\cite{monz11a}, currently the record for quantum registers, would require about $150$ days, although $100$ measurement outcomes can be collected for a single setting in only about three seconds. In photonic experiments this scaling is even worse because count rates are typically much lower, \eg, in recent eight-photon experiments~\cite{yao11a,huang11a} a coincidence of single click events occurs only at the order of minutes, hence it would require about seven years to collect an adequate data set. This directly shows that more sophisticated tomography techniques are mandatory.

New tomography protocols equipped with better scaling behaviour exploit the idea that the measurement scheme is explicitly optimized only for particular kinds of states rather than for all possible ones. If the true unknown state lies within this designed target class then full information about the state can be obtained with much less effort, and if the underlying density operator is not a member then a certificate signals that tomography is impossible in this given case. Recent results along this direction include tomography schemes designed for states with a low rank~\cite{gross10a,gross11a,flammia12a}, particularly for important low rank states like matrix product~\cite{cramer10a} or multi-scale entanglement renormalization ansatz states~\cite{landon_cardinal12a}. Other schemes include a tomography scheme based on information criteria~\cite{yin11a} or---the topic of this manuscript---states with permutation invariance~\cite{toth10a}.

However it needs to be stressed that in any real experiment all these tomography schemes must cope with another, purely statistical challenge: Since only a finite number of measurements can be carried out in any experiment one cannot access the true probabilities predicted by quantum mechanics $p_k=\tr(\rho_{\rm true} M_k)$, operator $M_k$ describing the measurements, but merely relative frequencies $f_k$. Though these deviations might be small, the approximation $f_k\approx p_k$ causes severe problems in the actual state reconstruction process, \ie, the task to determine the density operator from the observed data. If one na\"ively uses the frequencies according to Born's rule $f_k=\tr(\hat \rho_{\rm lin} M_k)$ and solves for the unknown operator $\hat \rho_{\rm lin}$, then, apart from possible inconsistencies in the set of linear equations, the reconstructed operator $\hat \rho_{\rm lin} \! \not\geq 0$ is often not a valid density operator anymore because some of its eigenvalues are negative. Hence in such cases this reconstruction called linear inversion delivers an unreasonable answer. It should be kept in mind that inconsistencies can also be due to systematic errors, \eg, if the true measurements are aligned wrongly relative to the respective operator representation~\cite{moroder12a,rosset12a}, but such effects are typically ignored.  

Therefore, statistical state reconstruction relies on other principles than linear inversion. In general, these methods require the solution of a non-linear optimization problem, which is much harder to solve than just a system of linear equations. For large system sizes this becomes, besides the exponential scaling of the number of settings and outcomes, a second major problem, again due to the exponential scaling of the number of parameters of the density operator. In fact for the current tomography record of eight ions in a trap \cite{haeffner05a}, this actual reconstruction took even longer than the experiment itself (one week versus a couple of hours). Hence, feasible quantum state tomography schemes for large systems must, in addition to an efficient measurement procedure, possess also an efficient state reconstruction algorithm, otherwise they are not scalable. 

In this paper we develop a scalable reconstruction algorithm for the proposed permutationally invariant tomography scheme~\cite{toth10a}. It works for common reconstruction principles, including, among others, maximum likelihood and least squares methods. This scheme becomes possible once more by taking advantage of the particular symmetry of this special kind of states, which provides an efficient and operational way to store, characterize and even process those states. This method enables a large dimension reduction in the underlying optimization problem such that it gets into the feasible regime. The final low dimensional optimization is performed via non-linear convex optimization which offers great advantages in contrast to commonly used numerical routines, in particular regarding numerical stability and accuracy. Already a first prototype implementation of this algorithm allows state reconstruction for $20$ qubits in a few minutes on a standard desktop computer.  

The outline of the manuscript is as follows: Section~\ref{sec:background} summarizes the background on permutationally invariant tomography and on statistical state reconstruction. The key method is explained in Sec.~\ref{sec:method} and highlighted via examples in Sec.~\ref{sec:examples}, that are generated by our current implementation. Section~\ref{sec:details} collects all the technical details: the mentioned toolbox, more notes about convex optimization, additional information about the pretest or certificate and the measurement optimization, both addressed for large qubit numbers. Finally, we conclude and provide an outlook on further directions in Sec.~\ref{sec:conclusion}.

\section{Background}\label{sec:background}

\subsection{Permutationally invariant tomography}\label{sec:recap_pitomo}

Permutationally invariant tomography has been introduced as a scalable reconstruction protocol for multi-qubit systems in Ref.~\cite{toth10a}. It is designed for all states of the system that remain invariant under all possible interchanges of its different particles. Abstractly such a permutationally invariant state $\rho_{\rm PI}$ of $N$ qubits can be expressed in the form, 
\begin{equation}
\label{eq:PI-state}
\rho_{\rm PI} = \left[ \:\rho \:\right]_{\rm PI} = \frac{1}{N!} \sum_{p\in S_N} V(p)\: \rho \: V(p)^\dag,
\end{equation}
where $V(p)$ is the unitary operator which permutes the $N$ different subsystems according to the particular permutation $p$ and the summation runs over all possible elements of the permutation group $S_N$. Many important states, like Greenberger-Horne-Zeilinger states or Dicke states fall within this special class. 

As shown in Ref.~\cite{toth10a}, full information of such states can be obtained by using in total ${N+2 \choose N} =(N^2+3N+2)/2$ different local binary measurement settings, while for each setting only the count rates of $(N+1)$ different outcomes need to be registered. This finally leads to a cubic scaling in contrast to the exponential scaling of standard tomography schemes. 

The measurement strategy that attains this number runs as follows: Each setting is described by a unit vector $\hat a \in \mathbbm{R}^3$ which defines associated eigenstates $\ket{i}_a$ of the trace-less operator $\hat a \cdot \vec \sigma$. Each party locally measures in this basis and registers the outcomes ``$0$'' or ``$1$'' respectively. The permutationally invariant part can be reconstructed from the collective outcomes, \ie, only the number of ``$0$'' or ``$1$''  results at the different parties matters but not the individual site information. The corresponding coarse-grained measurements are given by
\begin{eqnarray}
\label{eq:PI-measurements}
M_k^a&=& \sum_{p^\prime} V(p^\prime) \ket{0}_a\!\bra{0}^{\otimes k} \otimes  \ket{1}_a\!\bra{1}^{\otimes N-k} V(p^\prime)^\dag, \\
& = & {N \choose k} \left[ \ket{0}_a\!\bra{0}^{\otimes k} \otimes  \ket{1}_a\!\bra{1}^{\otimes N-k}\right]_{\rm PI}
\end{eqnarray}
with $\: k=0,\dots,N,$ and where the summation $p^\prime$ is over all permutations that give distinct terms. In total one needs the above stated number of different settings $\hat a$. These settings can be optimized in order to minimize the total variance which provides an advantage in contrast to random selection. 

In addition to this measurement strategy there is also a pretest which estimates the ``closeness'' of the true, unknown state with respect to all permutationally invariant states from just a few measurement results~\cite{toth10a}. This provides a way to test in advance whether permutationally invariant tomography is a good method for the unknown state. 

Restricting to the permutationally invariant part of a density operator has already been discussed in the literature; for example for spins in a Stern-Gerlach experiment~\cite{dariano03a} or in terms of the polarization density operator~\cite{karassiov05a,adamson07a,adamson08a}. Here, due to the restricted class of possible measurements, only the permutationally invariant part of, in principle distinguishable, particles is accessible~\cite{adamson07a,adamson08a}. This is a strong conceptual difference to permutationally invariant tomography where one intentionally constrains itself to this invariant part. Nevertheless it should be emphasized that the employed techniques are similar.

\subsection{Statistical state reconstruction}

Since standard linear inversion of the observed data typically results in unreasonable estimates as explained in the introduction, one employs other principles for actual state reconstruction. In general one uses a certain fit function $F(\rho)$ that penalizes deviations between the observed frequencies $f_k$ and the true probabilities predicted by quantum mechanics $p_k(\rho)=\tr(\rho M_k)$ if the state of the system is $\rho$. The reconstructed density operator $\hat \rho$ is then given by the (often unique) state that minimizes this fit function,
\begin{equation}
\label{eq:struct_srecon}
\hat \rho = \arg \:  \min_{\rho \geq 0} F(\rho),
\end{equation} 
hence the reconstructed state is precisely the one which best fits the observed data. Since the optimization explicitly restricts to physical density operators this assures validity of the final estimate $\hat \rho$ in contrast to linear inversion. Depending on the functional form of the fit function different reconstruction principles are distinguished. A list of the most common choices is provided in Tab.~\ref{tab:reconF}.

\begin{table}[ht]
\begin{center}
\begin{tabular}{lc}
\hline \hline \hline
Reconstruction principle \T\B& Fit function $F(\rho)$ \\ 
\hline
\T\B Maximum Likelihood~\cite{hradil97a} & $-\sum_k f_k \log[ p_k(\rho)]$ \\
\T\B Least Squares~\cite{langford07a} & $\sum_k w_k [ f_k - p_k(\rho)]^2$, $\:w_k > 0$ \\
\T\B Free Least Squares~\cite{james01a} & $\sum_k 1/p_k(\rho)[ f_k - p_k(\rho)]^2$ \\
\hline
\T\B Hedged Maximum Likelihood~\cite{blume10a} &$-\sum_k f_k \log[ p_k(\rho)]-\beta \log[\det(\rho)]$, $\:\beta > 0$ \\
\hline \hline \hline
\end{tabular}
\label{tab:reconF}
\caption{Common reconstruction principles and their corresponding fit functions $F$ used in the optimization given by Eq.~(\ref{eq:struct_srecon}); see text for further details.}
\end{center}
\end{table} 

The presumably best-known and most often employed method is called maximum likelihood principle~\cite{hradil97a}. Given a set of measured frequencies $f_k$ the maximum likelihood state $\hat \rho_{\rm ml}$ is exactly the one with the highest probability to generate these data. Other common fit functions, usually employed in photonic state reconstruction, are different variants of least squares~\cite{james01a,langford07a} that originate from the likelihood function using Gaussian approximations for the probabilities. There this is often also called maximum likelihood principle but we  distinguish these, indeed different, functions here explicitly. Typically the weights in the least squares function are set to be $w_k=1/f_k$ because $f_k$ represents an estimate of the variance in a multinomial distribution, cf. the free least squares principle. However, this leads to a strong bias if the counts rates are extremal, \eg, if one of the outcomes is never observed this method naturally leads to difficulties. A method to circumvent this is given by the free least squares function~\cite{james01a} or using improved error analysis for rare events. Let us stress that all these principles have the property that if linear inversion delivers a valid estimate $\hat \rho_{\rm lin} \geq 0$, it is also the estimate given by these reconstruction principles~\footnote{For the least squares fit functions this follows because $F=0$ in this case and clearly $F\geq 0$ for those functions. In the case of the likelihood it follows from positivity of the classical relative entropy between probability distributions.}. 

Finally, hedged maximum likelihood~\cite{blume10a} represents a method that circumvents low rank state estimates. Via this one obtains more reasonable error bars using parametric bootstrapping methods~\cite{bootstrapping}; for other error estimates we refer to the recently introduced confidence regions for quantum states~\cite{christandl11a,blume-kohout12a}. In principle many more fit functions are possible, like generic loss functions~\cite{mood}, but considering these is out of the scope of this work.

\section{Method}\label{sec:method}
From the previous it is apparent that permutationally invariant state reconstruction requires the solution of 
\begin{equation}
\label{eq:PI_staterecon}
\hat \rho_{\rm PI} = \arg \:  \min_{\rho_{\rm PI} \geq 0}  F(\rho_{\rm PI}),
\end{equation}
for the preferred fit function. This large-scale optimization becomes feasible along the following lines:

First, one reduces the dimensionality of the underlying optimization problem because one cannot work with full density operators anymore. This requires an operational way to characterize permutationally invariant states $\rho_{\rm PI} \geq 0$ over which the optimization is performed, and additionally demands an efficient way to compute probabilities $p_k(\rho_{\rm PI})$ which appear in the fit function. Second, one needs a method to perform the final optimization. We employ convex optimization for this task.

\subsection{Reduction of the dimensionality}

This reduction relies on an efficient toolbox to handle permutationally invariant states, which exploits the particular symmetry. Here we explain this method and the final structure; for more details see Sec.~\ref{sec:details_reduction}. These techniques are well-proven and established; we employ and adapt them here for the permutationally invariant tomography scheme such that we finally reach state reconstruction of larger qubits.

The methods of this toolbox are obtained via the concept of spin coupling that describes how individual, distinguishable spins couple to a combined system if they become indistinguishable. Since we deal with qubits we only need to focus on spin-$1/2$ particles. In the simplest case, two spin-$1/2$ particles can couple to a spin-$1$ system, called the triplet, if both spins are aligned symmetric, or to a spin-$0$ state, the singlet, if the spins are aligned anti-symmetric. Abstractly, this can be denoted as $\mathbbm{C}^2 \otimes \mathbbm{C}^2 = \mathbbm{C}^3 \oplus \mathbbm{C}^1$. If one considers now three spins, then of course all spins can point in the same direction giving a total spin-$3/2$ system. There is also to a spin-$1/2$ system possible if two particles form already a spin-$0$ and the remaining one stays untouched. This can be achieved however by more than one possibility, in fact by two inequivalent choices~\footnote{More precisely, all states of the form $\ket{\psi}=V(p)\ket{\psi^-}\otimes \ket{0}$ with $p$ being any possible permutation, are states of total spin $j=1/2$ and projection $m=1/2$ to the collective spin operators $J_i=\sum_{n=1}^3 \sigma_{i;n}/2$, $\sigma_{i;n}$ being the corresponding Pauli operator on qubit  $n$. However as can be checked these states only span a $2$-dimensional subspace.}, and is expressed by $\mathbbm{C}^2 \otimes \mathbbm{C}^2 \otimes \mathbbm{C}^2 = \mathbbm{C}^4 \oplus (\mathbbm{C}^2 \otimes \mathbbm{C}^2)$. 

This scheme can be extended to $N$ spin-$1/2$ particles to obtain the following decomposition of the total Hilbert space,
\begin{equation}
\mathcal{H}=(\mathbbm{C}^2)^{\otimes N} = \bigoplus_{j=j_{\rm min}}^{N/2} \mathcal{H}_j \otimes \mathcal{K}_j.
\end{equation}
where the summation runs over different total spin numbers $j=j_{\rm min},j_{\rm min}+1,\dots, N/2$ starting from $j_{\rm min}\in \{0,1/2\}$ depending on whether $N$ is even or odd. Here, $\mathcal{H}_j$ are called the spin Hilbert spaces with dimensions $\dim(\mathcal{H}_j)=2j+1$, while $\mathcal{K}_j$ are referred as multiplicative spaces that account for the different possibilities to obtain a spin-$j$ state. They are generally of a much larger dimension, cf.~Eq.~(\ref{eq:dimK}). 

Permutationally invariant states have a simpler form on this Hilbert space decomposition, namely
\begin{equation}
\label{eq:PI-state1}
\rho_{\rm PI} = \bigoplus_{j=j_{\rm min}}^{N/2} p_j \rho_j \otimes \frac{\mathbbm{1}}{\dim(\mathcal{K}_j)},
\end{equation}
with density operators $\rho_j$ called spin states and according probabilities $p_j$. Thus a permutationally invariant density operator only contains non-trivial parts on the spin Hilbert spaces while carrying no information on the multiplicative spaces. Note further that there are no coherences between different spin states. This means that any permutationally invariant state can be parsed into a block structure as schematically depicted in Fig.~\ref{fig:PIstate}.~The main diagonal is built up by unnormalized spin states $\tilde \rho_j=p_j \rho_j / \dim(\mathcal{K}_j)$, which appear several times, the number being equal to the dimension of the corresponding multiplicative space.~This block-decomposition represents a natural way to treat permutationally invariant states and has for example been employed already in the aforementioned related works of permutationally invariant tomography~\cite{dariano03a,karassiov05a,adamson07a,adamson08a} but also in other contexts~\cite{cirac99a,demkowicz05}.
\begin{figure}[ht]
\begin{center}
\includegraphics[scale=0.50]{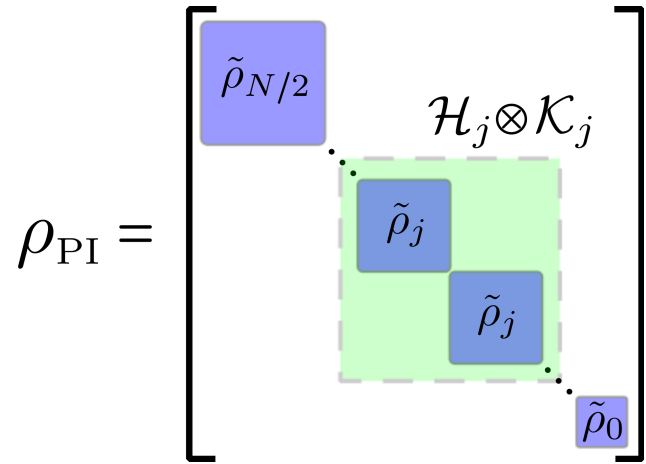}
\caption{Block decomposition for a generic permutationally invariant state as given by Eq.~(\ref{eq:PI-state1}) with $\tilde \rho_j=p_j \rho_j / \dim(\mathcal{K}_j)$.}
\label{fig:PIstate}
\end{center}
\end{figure}

This structure shows that if we work with permutationally invariant states we do not need to consider the full density operator but rather that it is sufficient to deal only with this ensemble of spin states. Therefore we identify from now on
\begin{equation}
\rho_{\rm PI} \Longleftrightarrow p_j \rho_j,\:j=j_{\rm min},j_{\rm min}+1,\dots, N/2.
\end{equation} 
This provides already an efficient way to store and to visualize such states. More importantly, it also enables an operational way to characterize valid states, since any permutationally invariant operator $\rho_{\rm PI}$ represents a true state if and only if all these spin operators $\rho_j$ are density operators and $p_j$ a probability distribution. This is in contrast to the generalized Bloch vector employed in the original proposal of permutationally invariant tomography~\cite{toth10a} given by Eq.~(\ref{eq:gen_bloch_vec}), which also provides efficient storage and processing of permutationally invariant states, but where Bloch vectors of physical states are not as straightforward to characterize.

Via this identification one can demonstrate once more the origin of the cubic scaling of the permutationally invariant tomography scheme. The largest spin state is of dimension $N+1$ which requires parameters on the order of $N^2$ for characterization. Since one has on the order of $N$ of these states this shows results in a cubic scaling.

Fixing the ensemble of spin states as parametrization it is now required to obtain an efficient procedure to compute the probabilities $p_k^a(\rho_{\rm PI})$ for the optimized measurement scheme. This is achieved as follows: First let us stress that a similar block decomposition as given by Eq.~(\ref{eq:PI-state1}) holds for all permutationally invariant operators. Hence also the measurements $M_k^a$ given by Eq.~(\ref{eq:PI-measurements}) can be cast into this form. Using the convention
\begin{equation}
\label{eq:PI-meas}
M_k^a = \bigoplus_{j=j_{\rm min}}^{N/2} M_{k,j}^a \otimes \mathbbm{1}
\end{equation}
leads to 
\begin{equation}
\label{eq:prob_PI}
\tr(\rho_{\rm PI} M_k^a) = \sum_{j=j_{\rm min}}^{N/2} p_j \tr(\rho_j M_{k,j}^a).
\end{equation}
Therefore the problem is shifted to the computation of the spin-$j$ operators $M_{k,j}^a$ for each setting $\hat a$. As we show in Prop.~\ref{prop:measurements} below, using the idea that measurements can be transformed into each other by a local operation $U_a\ket{i} = \ket{i}_a$ this provides the relation
\begin{equation}
M_{k,j}^a= W^a_j M_{k,j}^{e_3} W^{a,\dag}_j.
\end{equation}
Here $M_{k,j}^{e_3}$ corresponds to the measurement in the standard basis (that need to be computed once) and $W_j^a$ is a unitary transformation determined by the rotation $U_a$. This connection is given by
\begin{equation}
U_a = \exp(-i \sum_l t_l \sigma_l/2) \Longrightarrow W_j^a = \exp(-i \sum_l t_l S_{l,j})
\end{equation}
where $S_{l,j}$ stands for the spin operators in dimension $2j+1$. This finally provides the efficient way to compute probabilities.

\subsection{Optimization}

As a second step one still needs to cope with the optimization itself. Although there are different numerical routines for statistical state reconstruction like maximum likelihood \cite{hardil04a} or least squares~\cite{james01a,reimpell_thesis}, we prefer non-linear convex optimization \cite{cobook} to obtain the final solution. Quantum state reconstruction problems are known to be convex~\cite{reimpell_thesis,kosut04a}, but convex optimization has hardly been used for this task. However, convex optimization possesses several advantages: First of all it is a systematic approach that works for any convex fit function, including maximum likelihood and least squares. In contrast to other algorithms such as the 
fixed-point algorithm proposed in Ref.~\cite{hardil04a} it gives a precise stopping condition via an appropriate error control (see, however Ref.~\cite{glancy12a}) and still exploits all the favourable, convex, structure in comparison to re-parametrization ideas as in Ref.~\cite{james01a}. Moreover it is guaranteed to find the global optimum and the obtained accuracy is typically much higher than with other methods. 

Quantum state reconstruction as defined via Eq.~(\ref{eq:struct_srecon}) can be formulated as a convex optimization problem as follows: All fit functions listed in Tab.~\ref{tab:reconF} are convex on the set of states. Via a linear parametrization of the density operator $\rho(x)= \mathbbm{1}/\dim(\mathcal{H}) + \sum x_i B_i$, using an appropriate operator basis $B_i$ such that normalization is fulfilled directly, the required optimization problem becomes
\begin{eqnarray}
\label{eq:convex_opti}
\min_x && F[\rho(x)] \\ 
\nonumber
\textrm{s.t.}&& \rho(x) = \frac{\mathbbm{1}}{\dim(\mathcal{H})}+ \sum_i x_i B_i \geq 0,
\end{eqnarray}
with a convex objective function $F(x)=F[\rho(x)]$ and a linear matrix inequality as constraint, \ie, precisely the structure of a non-linear convex optimization problem~\cite{cobook}. For permutationally invariant states one uses $\rho(x)=\oplus_j \bar \rho_j(x)$ with $\bar \rho_j=p_j \rho_j$ by using an appropriate block-diagonal operator basis $B_i$; therefore we continue this discussion with the more general form.

The optimization given by Eq.~(\ref{eq:convex_opti}) can be performed for instance with the help of a barrier function~\cite{cobook}~\footnote{Let us stress that both least squares options can be parsed into a simpler convex problem, called semidefinite program, as for instance shown in Ref.~\cite{reimpell_thesis,langford07a}, but that this does not work with the true maximum likelihood function to our best knowledge.}. Rather than considering the constrained problem one solves the unconstrained convex task given by
\begin{equation}
\label{eq:uncon_opti}
\min_x F[\rho(x)] - t \log[\det\rho(x)],
\end{equation}
where the constraint is now directly included in the objective function. This so-called barrier term acts precisely as its name suggests: If one tries to leave the strictly feasible set, \ie, all parameters $x$ that satisfy $\rho(x) > 0$, one always reaches a point where at least one of the eigenvalues vanishes. Since the barrier term is large within this neighbourhood, in fact singular at the crossing, it penalizes points close to the border and thus ensures that one searches for an optimum well inside the region where the constraint is satisfied. The penalty parameter $t>0$ plays the role of a scaling factor. If it becomes very small the effect of the barrier term becomes negligible within the strictly feasible set and only remains at the border. Therefore a solution of Eq.~(\ref{eq:uncon_opti}) with a very small value of $t$ provides an excellent approximation to the real solution. As shown in Sec.~\ref{sec:copti_details} this statement can be made more precise by 
\begin{equation}
\label{eq:slackness}
F[\rho(x_{\rm sol}^t)] - F[\rho(x_{\rm sol})] \leq t \dim(\mathcal{H})
\end{equation}
which follows from convexity and which relates the true solution $x_{\rm sol}$ of the original problem given by Eq.~(\ref{eq:convex_opti}) to the solution $x_{\rm sol}^t$ of the unconstrained problem with penalty parameter $t$. This condition represents the above mentioned error control and serves as a stopping condition, \ie, as a quantitative error bound for a given $t$. Note that for permutationally invariant tomography $\dim(\mathcal{H})$ is not the dimension of the true $N$-qubit Hilbert space but instead the dimension of $\rho(x)=\oplus_j \bar \rho_j(x)$, \ie, $\sum_{j=j_{\rm min}} (2j+1)=(N+1)(N+2j_{\rm min}+1)/4$ which increases only quadratically.

Small penalty parameters are approached by an iterative process: For a given starting point $x_{\rm start}^{n}$ and a certain value of the parameter $t_n>0$ one solves Eq.~(\ref{eq:uncon_opti}). Its solution will be the starting point for $x_{\rm start}^{n+1}=x_{\rm sol}^{n}$ for the next unconstrained optimization with a lower penalty parameter $t_{n+1}< t_{n}$. As starting point for the first iteration we employ $t_0=1$ and the point $x_{\rm start}^{0}$ corresponding to the totally mixed state. This procedure is repeated until one has reached small enough penalty parameters. The penalty parameter is decreased step-wise. Then each unconstrained problem can be solved very efficiently since one starts already quite close to the true solution.

Let us point out that via the above mentioned barrier method one additionally obtains solutions to the hedged state reconstruction with $\beta=t$ since the unconstrained problem given by Eq.~(\ref{eq:uncon_opti}) is precisely the fit function of the hedged version of Tab.~\ref{tab:reconF}. 

Finally for comparative purposes we also like to mention the iterative fixed-point algorithm of Ref.~\cite{rehavcek01a}; for a modification see Ref.~\cite{rehacek07a}. It is designed for maximum likelihood estimation and is straightforward to implement since it only requires matrix multiplication, however, it has deficits regarding control and accuracy. For permutationally invariant tomography the algorithm can be adapted as follows: Given a valid iterate $\rho_{\rm PI}^n$ characterized by the ensemble of spin states $\bar \rho^n_j=p_j^n\rho_j^n$ one evaluates the probabilities $p_k^a(\rho_{\rm PI}^n)$ using Eq.~(\ref{eq:prob_PI}). Next, one computes the operators 
\begin{equation}
R^n_j = \sum_{a,k} \frac{f_k^a}{p_k^a(\rho_{\rm PI}^n)} M_{k,j}^a, 
\end{equation}
which determine the next iterate $\bar \rho^{n+1}_j= R^n_j \bar \rho_j^n R^{n\dag}_j/\mathcal{N}$ with $\mathcal{N}=\sum_j \tr(R^n_j \bar\rho_j^n R^{n\dag}_j)$. This iteration is started for example from the totally mixed state and repeated until a sufficiently good solution is obtained.

\section{Examples}\label{sec:examples}

The two methods from the previous section are employed in a prototype implementation under \textsc{MATLAB}, which already enables state reconstruction of about $20$ qubits on a standard desktop computer. 

The current algorithm is tested along the following lines: For a randomly generated permutationally invariant state $\rho_{\rm PI}^{\rm true}$ we compute the true probabilities $p_{k,\rm true}^a$ for the chosen measurement settings. Rather than sampling we set the observed frequencies equal to this distribution, \ie, $f_k^a=p_{k,\rm true}^a$. In this way linear inversion would return the original state, hence also each other reconstruction principle from Tab.~\ref{tab:reconF} has this state as solution. We now start the algorithm and compare the trace distance between the analytic solution $\rho_{\rm PI}^{\rm true}$ and the state after $n$ iterations $\rho_{\rm PI}^n$. This distance $\frac{1}{2}\tr|\rho_{\rm PI}^{\rm true} - \rho_{\rm PI}^n|$ quantifies the probability with which the two states, the true analytic solution and the iterate after $n$ steps in the algorithm, could be distinguished~\cite{nielsen_chuang}.

\begin{figure}[ht]
\begin{center}
\includegraphics[angle=-90,scale=0.51]{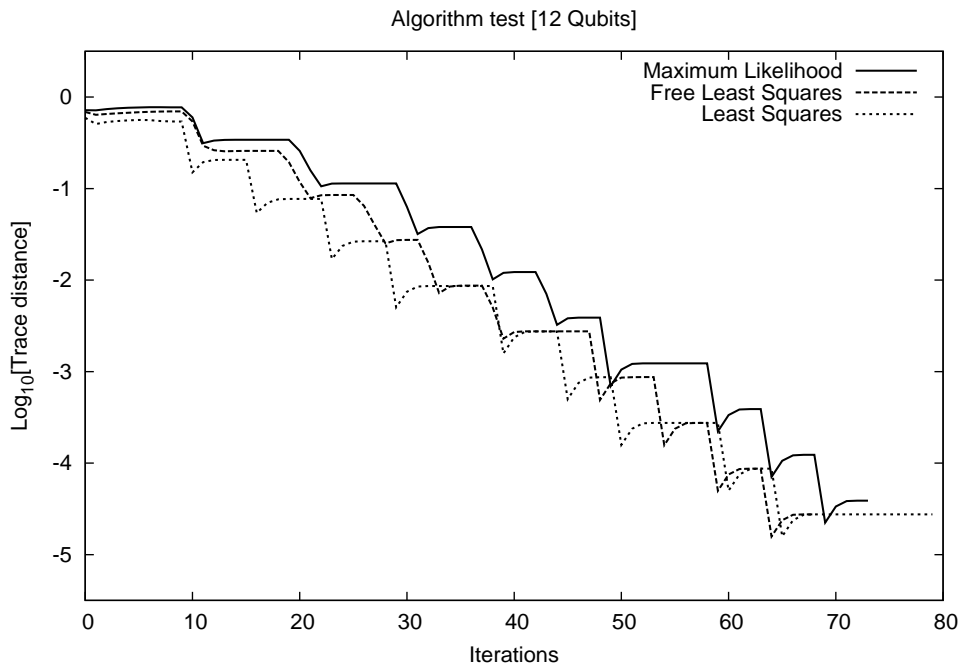}
\caption{Trace distance between the analytic solution and the estimate after $n$ algorithm steps for the three most common reconstruction principles from Tab.~\ref{tab:reconF}.}
\label{fig:example1}
\end{center}
\end{figure}
A typical representative of this process is depicted in Fig.~\ref{fig:example1} for $12$ qubits using optimized settings. The randomly generated state $\rho_{\rm PI}^{\rm true}$ was chosen to lie at the boundary of the state space since such rank-reduced solutions better resemble the case of state reconstruction of real data. More precisely, each spin state of the true density operator is given by a pure state $\rho^{\rm true}_j=\ket{\psi_j}\bra{\psi_j}$ chosen according to the Haar measure, while the $p_j$ are selected by the symmetric Dirichlet distribution with concentration measure $\alpha=1/2$~\cite{zyczkowski01}. As apparent the algorithm behaves similar for all three reconstruction principles and rapidly obtains a good solution after about $70$ iterations. The steps in this plot are points where the penalty parameter is reduced by a factor of $10$ starting from $t=1$ and decreased down to $t=10^{-10}$. The slight rise after these points comes from the fact that we plot the trace distance and not the actual function (fit-function plus penalty term) that is minimized. 

\begin{figure}[hb]
\begin{center}
\includegraphics[angle=-90,scale=0.51]{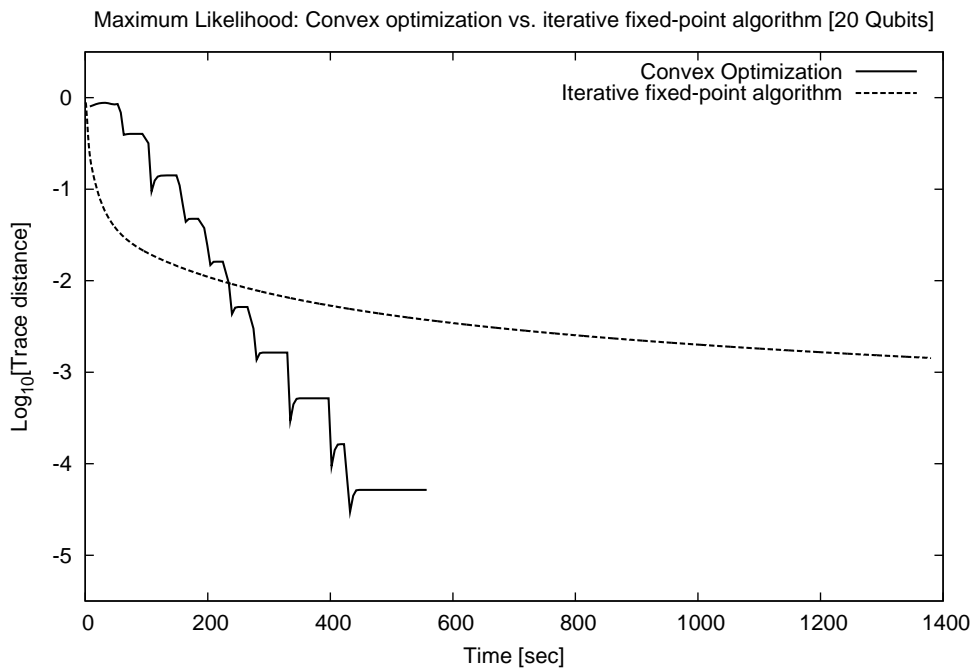}
\caption{Comparison of the maximum likelihood principle between convex optimization and the iterative fixed-point algorithm for the described testing procedure with respect to accuracy and algorithm time.}
\label{fig:example2}
\end{center}
\end{figure}
Figure~\ref{fig:example2} shows a similar comparison for the maximum likelihood reconstruction of $20$ qubits but now plotted versus algorithm time~\footnote{All simulations were performed on an Intel Core i5-650, 3.2 Ghz, 8 GB RAM using MATLAB 7.12.}. For comparison we include the performance of the iterative fixed-point algorithm, which requires much more iterations in general ($3000$ in this case vs.~about $90$ for convex optimization). Let us emphasise that a similar behaviour between these two algorithms appears also for smaller qubit numbers. As one can see, convex optimization delivers a faster and in particular more accurate solution. In contrast, the iterative fixed-point algorithm shows a bad convergence rate although it initially starts off better. This was one of the main reasons for us to switch to convex optimization.

The current algorithm times of this test are listed in Tab.~\ref{tab:current_performance} which are averaged over $50$ randomly generated states. Thus already this prototype implementation enables state reconstruction of larger qubit numbers in moderate times. The small time difference between reconstruction principles is because least squares as a quadratic fit function provides some advantages in the implementation. More details about this difference are given in Sec.~\ref{sec:details}.

\begin{table}
\begin{center}
\begin{tabular}{ccccc}
\hline \hline \hline
\T\B \phantom{Tomography protocol     }  & $N=8$ & $N=12$ & $N=16$ &$N=20$ \\
\hline \hline \hline
\T\B Maximum Likelihood  & & & & \\
\hline
\T\B Algorithm test  & $8.5 \sec$ & $47 \sec$ & $2.7 \min$  &$9.2 \min$\\
\T\B Simulated experiment & $9.2 \sec$ & $48 \sec$ & $2.9 \min$  &$9.3 \min$ \\
\hline \hline \hline
\T\B Least Squares  & & & & \\
\hline
\T\B Algorithm test  & $8.4 \sec$ & $39 \sec$ & $2.5 \min$  &$6 \min$\\
\T\B Simulated experiment & $9.2 \sec$ & $43 \sec$ & $2.7 \min$  &$6.7 \min$ \\
\hline \hline \hline
\end{tabular}
\caption{Current performance of the convex optimization algorithm on the described test procedure and on frequencies from simulated experiments; free least squares provides similar results to the maximum likelihood principle.}
\label{tab:current_performance}
\end{center}
\end{table}

Table~\ref{tab:current_performance} also contains the algorithm times for reconstructions using simulated frequencies $f_k^a=n_k^a/N_{\rm r}$. For each setting they are deduced from the count rates $n_k^a$ sampled from a multinomial distribution using the true event distribution $p_{k,\rm true}^a$ and $N_{\rm r}=1000$ repetitions. The true probabilities correspond to the same states as already employed in the algorithm test. From the table one observes that state reconstruction for data with count statistics requires only slightly more time than the algorithm test with the correct probabilities. We attribute this to the fact that a few more iterations are typically required in order to achieve the desired accuracy.  

Finally let us perform the reconstruction of a simulated experiment of $N=14$ qubits. Suppose that one intends to create a Dicke state $\ket{D_{k,N}}$ as given by Eq.~(\ref{eq:dicke_basis}), but that the preparation suffers from some imperfections such that at best one can prepare states of the form
\begin{equation}
\label{eq:aimed_state}
\rho_{\rm dicke} = \sum_{k=0}^{N} {N \choose k} p^k (1-p)^{N-k} \ket{D_{k,N}}\bra{D_{k,N}},
\end{equation}
where $p=0.6$ characterizes some asymmetry. As the true state prepared in the experiment we now model some further imperfection in the form of an additional small misalignment $U^{\otimes N}$, with $U=\exp(-i \theta \sigma_y/2)$, $\theta=0.2$, and some permutationally invariant noise $\sigma_{\rm PI}$ (chosen via the aforementioned method but using Hilbert-Schmidt instead of the Haar measure), \ie,
\begin{equation}
\label{eq:prepared_state}
\rho_{\rm true}=0.6 \:U^{\otimes N} \rho_{\rm dicke} U^{\dag \otimes N}  + 0.4\: \sigma_{\rm PI}.
\end{equation}
The frequencies are obtained via sampling from the state given by Eq.~(\ref{eq:prepared_state}) using intentionally only $N_{\rm r}=200$ repetitions per setting (to see some differences). Finally, we reconstruct the state according to the maximum likelihood principle. Figure~\ref{fig:example} shows the tomography bar plots of one of these examples for the largest spin state $p_j\rho_{j}$, $j=N/2=7$ for both states. Though this state might be artificial this example should highlight once more that this state reconstruction algorithm works also for realistic data and for qubit sizes where clearly any non-tailored state reconstruction scheme would fail. Moreover, it demonstrates that the spin ensemble $p_j\rho_j$ represents a very convenient graphical representation of such states (compared to a $2^{14} \times 2^{14}$ matrix in this case).
\begin{figure}[h]
\begin{center}
\includegraphics[scale=1.1]{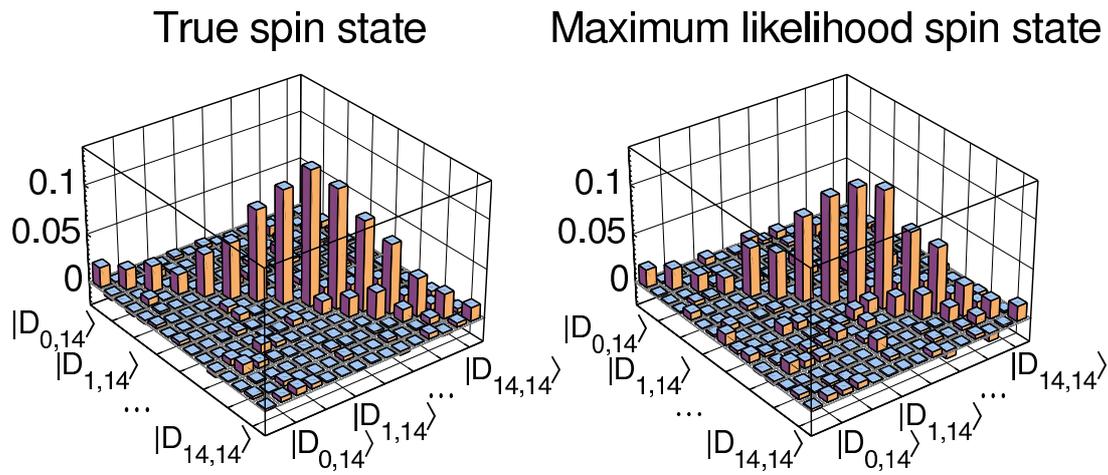}
\caption{The real part of the true and reconstructed (according to maximum likelihood) largest spin ensemble $p_j\rho_{j}$, $j=N/2=7$ using the optimal measurement setting. The basis is given by the Dicke basis $\ket{D_{k,14}}$, cf. Eq.~(\ref{eq:dicke_basis}).}
\label{fig:example}
\end{center}
\end{figure}

\section{Details}\label{sec:details}

\subsection{Reduction}\label{sec:details_reduction}

Let us first give more details regarding the reduction step. This starts by recalling a group theoretic result summarized in the next section, which is then used to show how the stated simplifications with respect to states and measurements are obtained.

\subsubsection{Background}\label{sec:schur-weyl}
Consider the following two unitary representations defined on the $N$ qubit Hilbert space: The permutation or symmetric group $V(p)$ which is defined by their action onto a standard tensor product basis by $V(p)\ket{i_1, \dots, i_N}=\ket{i_{p^{-1}(1)},\dots, i_{p^{-1}(N)}}$ according to the given permutation $p$, and the tensor product representation $W(U)=U^{\otimes N}$ of the special unitary group. A result known as the Schur-Weyl duality~\cite{simon,christandl06} states that the action of these two groups is dual, which means that the total Hilbert space can be divided into blocks on which the two representations commute. More precisely one has 
\begin{eqnarray}
(\mathbbm{C}^2)^{\otimes N} &=& \bigoplus_{j=j_{\rm min}}^{N/2} \mathcal{H}_j \otimes \mathcal{K}_j, \\
\label{eq:rep_symm}
V(p) &=& \bigoplus_{j=j_{\rm min}}^{N/2} \mathbbm{1} \otimes V_j(p), \\
\label{eq:rep_SU(2)}
W(U) &=& \bigoplus_{j=j_{\rm min}}^{N/2} W_j(U) \otimes \mathbbm{1}.
\end{eqnarray}
Here $V_j$ and $W_j$ are respective irreducible representations, and $j_{\rm min} \in \{0, 1/2\}$ depending on whether $N$ is even or odd. The dimensions of the appearing Hilbert spaces are $\dim(\mathcal{H}_j)=2j+1$ and 
\begin{equation}
\label{eq:dimK}
\dim(\mathcal{K}_j) = {N \choose N/2-j} - {N \choose N/2-j-1}
\end{equation}
for all $j<{N/2}$ and $\dim(\mathcal{K}_{N/2})=1$. Let us note that Eq.~(\ref{eq:rep_symm}) already ensures the block-diagonal structure of permutationally invariant operators, while Eq.~(\ref{eq:rep_SU(2)}) becomes important for the measurement computation.

A basis of the Hilbert space $\mathcal{H}_j \otimes \mathcal{K}_j$ is formed by the spin states $\ket{j,m,\alpha}=\ket{j,m} \otimes \ket{\alpha_j}$ with $m=-j,\dots,j$ and $\alpha_j=1,\dots,\dim(\mathcal{K}_j)$. These are obtained by starting with the states having the largest spin number $m=j$, which are given by
\begin{eqnarray}
\ket{j,j,1}&=&\ket{0}^{\otimes 2j} \otimes \ket{\psi^-}^{\otimes N-2j}, \\
\ket{j,j,\alpha}&=&\sum_p c_{j,p} V(p) \ket{j,j,1}, 
\end{eqnarray}
for all $\alpha \geq 2$. The coefficients $c_{j,p}$ must ensure that the states $\ket{j,j,\alpha}$ are orthogonal, otherwise their choice is completely free since the detailed structure of different $\alpha$'s is not important. The full basis is obtained by subsequently applying the ladder operator $J_-=\sum_{n=1}^N \sigma_{-;n}$ to decrease the spin number $m$. Here $\sigma_{-;n}$ refers to the operator with $\sigma_-=(\sigma_x-i \sigma_y)/2$ on the $n$-th qubit and identity on the rest. Thus in total the basis becomes 
\begin{equation}
\ket{j,m,\alpha} = \mathcal{N} J_-^{j-m} \ket{j,j,\alpha},
\end{equation}
with appropriate normalizations $\mathcal{N}$. Note that the subspace corresponding to the highest spin number $j=N/2$ is also called the symmetric subspace, which contains many important states like Greenberger-Horne-Zeilinger or Dicke states, which using the spin states read as
\begin{eqnarray}
\ket{\rm{GHZ}}\!\!&=&\!\frac{1}{\sqrt{2}}\! \left( \ket{0}^{\otimes N} + \ket{1}^{\otimes N}\right) \!=\!\frac{1}{\sqrt{2}}\! \left( \ket{N/2,N/2,1} + \ket{N/2,-N/2,1} \right)\!, \\
\label{eq:dicke_basis}
\ket{D_{k,N}}\!&=& \mathcal{N} \left[ \ket{1}^{\otimes k}\otimes \ket{0}^{\otimes N-k} \right]_{\rm PI}= \ket{N/2, N/2-k,1}.
\end{eqnarray}

\subsubsection{Permutationally invariant states and measurement operators}

Let us now employ this result in order to derive a generic form for permutationally invariant states; we give the proof for completeness.

\begin{proposition}[Permutationally invariant states]\label{prop:PIstate}
Any permutationally invariant state of $N$ qubits $\rho_{\rm PI}$ defined via Eq.~(\ref{eq:PI-state}) can be written as
\begin{equation}
\rho_{\rm PI} = \bigoplus_{j=j_{\rm min}}^{N/2} p_j \rho_j \otimes \frac{\mathbbm{1}}{\dim(\mathcal{K}_j)},
\end{equation}
hence it is fully characterized already by the ensemble $p_j \rho_j$. Moreover $\rho_{\rm PI}$ is a density operator if and only if all $\rho_j$ are density operators and $p_j$ a probability distribution. 
\end{proposition}

\begin{proof}
The proposition follows using the representation $V(p)$ given by Eq.~(\ref{eq:rep_symm}) in the definition of the states Eq.~(\ref{eq:PI-state}) and then applying Schur's lemma~\cite{simon,christandl06}. This lemma states that any linear operator $A$ from $\mathcal{K}_j$ to $\mathcal{K}_i$ which commutes with all elements $p$ of the group $V_i(p) A = A V_j(p)$ must either be zero if $i$ and $j$ label different irreducible representations or $A= c \mathbbm{1}$ if they are unitarily equivalent. Since $A_{\rm PI}=1/N! \sum_p V_i(p) A V_j(p)^\dag$ fulfils this relation one obtains
\begin{equation}
\frac{1}{N!} \sum_p V_i(p) A V_j(p)^\dag = \delta_{ij} \: \tr(A) \frac{\mathbbm{1}}{\dim(\mathcal{K}_j)}.
\end{equation}
The normalization can be checked taking the trace on both sides. Adding appropriate identities provides
\begin{equation}
\label{eq:help3}
\frac{1}{N!} \sum_p \mathbbm{1} \otimes V_i(p) B \mathbbm{1} \otimes V_j(p)^\dag = \delta_{ij} \: \tr_{\mathcal{K}_j}( B ) \otimes \frac{\mathbbm{1}}{\dim(\mathcal{K}_j)}
\end{equation} 
where $B$ should now be a linear operator from $\mathcal{H}_j \otimes \mathcal{K}_j$ to $\mathcal{H}_i \otimes \mathcal{K}_i$. 

Finally let $P_j$ denote the projector onto $\mathcal{H}_j\otimes \mathcal{K}_j$ and using Eq.~(\ref{eq:help3}) delivers
\begin{eqnarray}
\rho_{\rm PI} &=& \frac{1}{N!}\sum_p V(p) \rho V(p)^\dag = \sum_{i,i^\prime,j,j^\prime} \frac{1}{N!}\sum_p P_i V(p)P_{i^\prime} \rho P_j V(p)^\dag P_{j^\prime} \\ 
&=& \sum_{i,j}  \left\{ \frac{1}{N!} \sum_p P_i[\mathbbm{1} \otimes V_i(p)] P_i \rho_{\rm PI} P_{j}[\mathbbm{1} \otimes V_{j}(p)]^\dag P_{j}\right\}\\
&=& \sum_{i,j} P_i \left[ \delta_{ij} \: \tr_{\mathcal{K}_j}(P_i \rho_{\rm PI} P_{j}) \otimes \frac{\mathbbm{1}}{\dim(\mathcal{K}_j)} \right] P_j \\
&=& \bigoplus_j \tr_{\mathcal{K}_j}(P_j \rho_{\rm PI} P_j) \otimes \frac{\mathbbm{1}}{\dim(\mathcal{K}_j)},
\end{eqnarray} 
which provides the general structure. 

The state characterization part follows because positivity of a block-matrix is equivalent to positivity of each block.
\end{proof}

Next let us concentrate on the measurement part. Though the block decomposition follows already from the previous proposition, it is here more important to obtain an efficient computation of each measurement block for the selected setting. 

\begin{proposition}[Measurement operators]\label{prop:measurements}
The POVM elements $M_k^a$ as defined in Eq.~(\ref{eq:PI-measurements}) for any local setting $\hat a \in \mathbbm{R}^3$ can be decomposed as $M_k^a=\bigoplus_j M_{k,j}^a \otimes \mathbbm{1}$ with  
\begin{equation}
M_{k,j}^a=W_j(U_a) M^{e_3}_{k,j} W_j(U_a)^\dag.
\end{equation} 
The unitary is given by $W_j(U_a)=\exp(-i \sum_l t_l S_{l,j})$ using the spin operators $S_{l,j}$ in dimension $2j+1$, while the coefficients $t_l$ are determined by $U_a = \exp(-i \sum_l t_l \sigma_l/2)$ which satisfies $\hat a \cdot \vec \sigma= U_a \sigma_z U_{a}^\dag$. For the measurement in the standard basis $\hat a=\hat{e}_3$ one gets
\begin{equation}
M^{e_3}_{k,j}= \ket{j,N/2-k}\bra{j,N/2-k}
\end{equation}
if $-j \leq N/2-k \leq j$ and zero otherwise.
\end{proposition}

\begin{proof}
The basic idea is to consider the measurement in an arbitrary local basis $\hat a$ by a rotation followed by the collective measurement in the standard basis. The block decomposition is obtained as follows
\begin{eqnarray}
M_k^a &=& U_a^{\otimes N} M_k^{e_3} U_a^{^\dag \otimes N} = W(U_a) \Big[ \bigoplus_j  M^{e_3}_{k,j} \otimes \mathbbm{1} \Big] W(U_a)^\dag  \\
 &=& \bigoplus_j W_j(U_a) M^{e_3}_{k,j} W_j(U_a)^\dag \otimes \mathbbm{1}.
\end{eqnarray}
The first step holds because $U_a$ satisfies $\ket{i}_a\!\bra{i}=U_a \ket{i}\bra{i} U_a^\dag$, while the block decomposition of the standard basis measurement $M_k^{e_3}$ is employed afterwards. In the last part one uses the tensor product representation given by Eq.~(\ref{eq:rep_SU(2)}). 

Since one knows that $W_j$ is irreducible it can be uniquely written in terms of its Lie algebra representation $dW_j$ as $W_j(U_a)=W_j(e^{-i X}) = e^{-i dW_j(X)}$, which is given by the spin operators in this case, \ie, $dW_j(\sigma_l/2) = S_{l,j}$~\cite{hall}.  

Thus it is left to compute the measurement blocks $M^{e_3}_{k,j}$ for the standard basis. Note it is sufficient to evaluate $M_{k,j} \otimes \ket{1_j}\bra{1_j}$ such that one can employ the spin basis states $\ket{j,m,1}$ as introduced in Sec.~\ref{sec:schur-weyl}. At first note that $M^{e_3}_{k,j}$ exactly contains $k$ projections onto $\ket{0}$, while each basis state $\ket{j,m,1}$ possesses $(N/2+m)$ zeros. Therefore one obtains $M^{e_3}_{k,j} \ket{j,m,1} \propto \delta_{k,N/2+m} \ket{j,m,1}$. Since each POVM has to resolve the identity this is only possible if each $M_{k,j}^{e_3}$ is the stated rank-$1$ projector on the basis states.
\end{proof}

Finally one still needs to express $U_a=\exp(-i \sum_l t_l \sigma_l /2)$ for the chosen setting $\hat a \in \mathbbm{R}^3$. Since this can be related to a familiar rotation~\cite{hall} these coefficients can be expressed as $t_l = (\theta \hat n)_l$ via a rotation about an angle $\theta$ around the axis $\hat n$. Since this rotation should turn $\hat e_3$ into $\hat a$ its parameters are given by 
\begin{eqnarray}
\hat n &=& \frac{\hat e_3 \times \hat a}{\| \hat e_3 \times \hat a \|_2}, \\
\theta &=& \arccos(\hat e_3 \cdot \hat a),
\end{eqnarray}
and $\hat n = \hat e_1$ (or any other orthogonal vector) if $\hat a = \pm \hat e_3$.

\subsection{Convex optimization}\label{sec:copti_details}

In this part we collect some more details regarding the described convex optimization algorithm; for a complete coverage we refer to the book~\cite{cobook}.

Each unconstrained optimization given by Eq.~(\ref{eq:uncon_opti}) is solved via a damped Newton algorithm. The minimization of $f(x)=F[\rho(x)] - t \log[\det\rho(x)]$ is obtained by an iterative process. In order to determine a search direction at a given iterate $x^n$ one minimizes the quadratic approximation 
\begin{equation}
f(x^n+\Delta x) \approx f(x^n) + \nabla f(x^n)^T \Delta x + \frac{1}{2} \Delta x^T \nabla^2 f(x^n) \Delta x.
\end{equation}
This reduces to solving a linear set of equation called the Newton equation
\begin{equation}
\nabla^2 f(x^n) \Delta x_{\rm nt} = - \nabla f(x^n),
\end{equation}
which determines the search direction $\Delta x_{\rm nt}$. The steplength $s$ for the next iterate $x^{n+1}= x^{n} + s \Delta x_{\rm nt}$ is chosen by a backtracking line search. Here one picks the largest $s=\max_{k \in \mathbbm{N}}\beta^k$ with $\beta \in (0,1)$ such that the iterate stays feasible $\rho(x^{n+1})>0$ and that the function value decreases sufficiently, \ie, $f(x^{n+1}) \leq f(x^n) + \alpha s \nabla f(x^n)^T \Delta x_{\rm nt}$ with $\alpha \in (0,0.5)$. The process is stopped if one has reached an appropriate solution, which can be identified by $\| \nabla f(x^n) \|_2 \leq \epsilon$. If the initial point $x_{\rm start}$ is already sufficiently close to the true solution then the whole algorithm converges quadratically, \ie, the precision gets doubled at each step.

At this point let us give the gradient and Hessian of the appearing functions. For the barrier term $\psi(x) = - \log [ \det \rho(x)]$ restricted to the positive domain $\rho(x)= \mathbbm{1}/\dim(\mathcal{H})+\sum_i x_i B_i > 0$ one gets~\cite{cobook}
\begin{eqnarray}
\label{eq:gradient_barrier}
\frac{\partial \psi(x)}{\partial x_i}  &=& - \tr[ \rho(x)^{-1} B_i], \\
\label{eq:hessian_barrier}
\frac{\partial^2 \psi(x)}{\partial x_j\partial x_i} &=& \tr[\rho(x)^{-1} B_j \rho(x)^{-1} B_i].
\end{eqnarray}
Equation~(\ref{eq:hessian_barrier}) shows that the Hessian of the penalty term $\nabla^2 \psi(x)>0$ is positive definite, such that $\psi(x)$ is indeed convex. The derivatives of the preferred fit function can be computed directly. For instance using the likelihood function $F_{\rm ml}(x)= - \sum_k f_k \log p_k(x)$ with $p_k(x)=\tr[\rho(x) M_k]$ they read
\begin{eqnarray}
\frac{\partial F_{\rm ml}(x)}{\partial x_i} &=& - \sum_k \frac{f_k}{p_k(x)} \tr(B_i M_k), \\
\label{eq:hessian_maxlik}
\frac{\partial^2 F_{\rm ml}(x)}{\partial x_j\partial x_i} &=&\sum_k \frac{f_k}{p^2_k(x)} \tr(B_j M_k ) \tr(B_i M_k ).
\end{eqnarray}

The bottleneck of such an algorithm is the actual computation of the second derivatives. Although the expansion coefficients of each measurement $\tr(B_j M_k)$ can be computed in advance, it is still necessary to compute Eq.~(\ref{eq:hessian_maxlik}) anew at each point $x$ due to the dependence of $p_k(x)$. With respect to that the least squares fit function bears a great advantage since its Hessian is constant, \ie, $\partial_j\partial_i F_{\rm ls}(x) = 2 \sum_k w_k \tr(B_j M_k ) \tr(B_i M_k )$, such that one saves time on this part. 

At last let us comment on the optimality conditions, known as the Karush-Kuhn-Tucker conditions~\cite{cobook}. A given $x^\star$ is the global solution of the convex problem given by Eq.~(\ref{eq:convex_opti}) if and only if~\footnote{Sufficiency holds under the Slater regularity condition that demands a strictly feasible point $\rho(x)>0$, which naturally holds for state reconstruction problems.} there exists an additional Lagrange multiplier $Z^\star$ such that the pair $(x^\star,\Lambda^\star)$ satisfies
\begin{eqnarray}
\label{eq:kkt_gradient}
\frac{\partial}{\partial x_i} F(x^\star)-\tr[\Lambda^\star B_i] &=& 0, \;\forall i,\\
\label{eq:kkt_feasibility}
\Lambda^\star \geq 0, \;\;\rho(x^\star) \geq 0, \\ 
\label{eq:kkt_gap}
\tr[\Lambda^\star \rho(x^\star)] = 0. 
\end{eqnarray}
Given the solution $x^t_{\rm sol}$ of the corresponding unconstrained problem with penalty parameter~$t$ it follows from $\nabla f(x^t_{\rm sol})=0$ using Eq.~(\ref{eq:gradient_barrier}) that the gradient conditions are satisfied with $\Lambda_t = t \rho(x^t_{\rm sol})^{-1}$ being the Lagrange multiplier. This pair $(x_{\rm sol}^t,\Lambda_t)$ also satisfies Eq.~(\ref{eq:kkt_feasibility}); only the duality gap condition $\tr[\Lambda_t \rho(x^t_{\rm sol})] = t \dim(\mathcal{H}) > 0$ does not hold exactly. However this quantity appears in the following inequality
\begin{eqnarray}
F(x^t_{\rm sol}) - \tr[ \Lambda_t \rho(x^t_{\rm sol})] &=& \min_{x: \rho(x) \geq 0} F(x) - \tr[\Lambda_t \rho(x)] \\ &\leq& \min_{x: \rho(x) \geq 0} F(x) = F(x_{\rm sol}).
\end{eqnarray} 
Here one used that $x_{\rm sol}^t$ is the solution of $F(x)-\tr[\Lambda_t\rho (x)]$ because the gradient vanishes (and the solution is not at the border), and  $\tr[\Lambda_t \rho(x)] \geq 0$ for the inequality. This is the stated accuracy given by Eq.~(\ref{eq:slackness}) which relates the function value of $x^t_{\rm sol}$ to the true solution $x_{\rm sol}$.

\subsection{Additional tools}

\subsubsection{Optimization of measurement settings}\label{sec:optimized_settings}

Measurement settings, each described by a unit vector $\hat{a}_i \in \mathbbm{R}^3$ as explained in Sec.~\ref{sec:recap_pitomo}, are chosen to optimize a figure of merit characterizing how well a given permutationally invariant target state $\rho_{\rm tar}$ can be reconstructed. This is motivated as follows: A permutationally invariant state $\rho_{\rm PI}$ is uniquely described by its generalized Bloch vector~\cite{toth10a} defined as
\begin{equation}
\label{eq:gen_bloch_vec}
b_{klmn}=\tr(\left[ \sigma_x^{\otimes k} \otimes \sigma_y^{\otimes l} \otimes \sigma_z^{\otimes m} \otimes \mathbbm{1}^{\otimes n} \right]_{\rm PI} \rho_{\rm PI}) 
\end{equation}
with natural numbers satisfying $k+l+m+n=N$. Consequently, one possible figure of merit is given by the total error of all Bloch vector elements, \ie, more precisely by
\begin{equation}
\mathcal{E}^2_{\rm total}(\hat a_i,\rho_{\rm tar})= \sum_{k,l,m,n} {N \choose k,l,m,n} \; \mathcal{E}^2_{b_{klmn}}(\hat a_i,\rho_{\rm tar}).
\end{equation}
Here the multinomial coefficient weights the error of each Bloch vector by its number of appearance in a generic Pauli product decomposition.

The error of each Bloch vector element must now be related to the performed measurements. For that note that each Bloch vector element can be expressed as
\begin{equation}
b_{klmn} = \sum_i c_i^{klmn} \tr( [(\hat{a}_i \cdot \vec \sigma)^{\otimes N-n} \otimes \mathbbm{1}^{\otimes n}]_{\rm PI} \rho_{\rm tar})
\end{equation}
using appropriate coefficients $c_i^{klmn}$ and the expectation values of $[(\hat{a}_i \cdot \vec \sigma)^{\otimes N-n} \otimes \mathbbm{1}^{\otimes n}]_{\rm PI}$ which can be computed from the coarse-grained measurement outcomes $M_k^{a_i}$ of setting $\hat{a}_i$ as given by Eq.~(\ref{eq:PI-measurements}) using linear combinations. Assuming independent errors one gets
\begin{eqnarray}
\mathcal{E}^2_{b_{klmn}}(\hat a_i,\rho_{\rm tar}) = \sum_i c_i^{klmn} \mathcal{E}^2_{\rho_{\rm tar}} \left( [(\hat{a}_i \cdot \vec \sigma)^{\otimes N-n} \otimes \mathbbm{1}^{\otimes n}]_{\rm PI} \right).
\end{eqnarray}
The detailed form of the error expression $\mathcal{E}^2_{\rho_{\rm tar}} \left( [(\hat{a}_i \cdot \vec \sigma)^{\otimes N-n} \otimes \mathbbm{1}^{\otimes n}]_{\rm PI} \right)$ may depend on the actual physical realization, but we assume the following form
\begin{equation}
\mathcal{E}^2_{\rho_{\rm tar}} \left( [(\hat{a}_i \cdot \vec \sigma)^{\otimes N-n} \otimes \mathbbm{1}^{\otimes n}]_{\rm PI} \right) =K \Delta_{\rho_{\rm tar}} \left( [(\hat{a}_i \cdot \vec \sigma)^{\otimes N-n} \otimes \mathbbm{1}^{\otimes n}]_{\rm PI} \right),
\end{equation}
where $\Delta_{\rho}[A]=\tr(\rho A^2)-[\tr(\rho A)]^2$ is the standard variance and $K$ a state-independent factor. This form fits for example well to the common error model in photonic experiments where count rates are assumed to follow a Poissonian distribution. For more details on this derivation we refer to Ref.~\cite{toth10a}.

For large qubit numbers $N$ this optimization is carried out iteratively. Starting from randomly chosen measurement directions or from vectors which are uniformly distributed according to some frame potential~\cite{gross07a}, one searches for updates according to 
\begin{equation}
\hat{a}_i^\prime=\frac{p\hat{a}_i^n+(1-p)\hat{r}_i}{\| p\hat{a}_i^n+(1-p)\hat{r}_i \|}.
\end{equation}
Here $\hat{a}_i^n$ is the current iterate, $p<1$ a probability close to $1$ and $\hat{r}_i$ are randomly chosen unit vectors. If this new set of directions $\hat{a}_i^\prime$ leads to a smaller total error $\mathcal{E}^2_{{\rm total}}(\hat{a}^\prime_i,\rho_{\rm tar})$ than the previous set then this new measurement settings form the next iterate $\hat{a}^{n+1}_i=\hat{a}_i^n$, otherwise this procedure is carried out once more. This process is repeated until the total error does not decrease any more. This way of optimizing the measurements requires a method to compute the total error $\mathcal{E}^2_{{\rm total}}(\hat{a}_i,\rho_{\rm tar})$ for a given set of measurements $\hat{a}_i$. Using the generalized Bloch vector or the spin ensemble this computation can be carried out again efficiently for larger qubit numbers $N$. 

Though this algorithm is not proven to attain the true global optimum it is still a straightforward technique to obtain good settings. In the end this is often sufficient, recalling that the true unknown state can deviate from the target state and that one considers ``just'' an overall single figure of merit. For our simulations we always use the optimized settings for the totally mixed state. 

Regarding this point we finally like to add that if one does not employ the minimal number of measurement settings, but rather an over-complete set, \eg, four times as much settings but four times fewer measurements per setting, then the procedure is quite insensitive to the chosen measurement directions. Hence, in many practical situations the search for optimal directions might not even be necessary and randomly chosen measurement directions suffice equally well.

\subsubsection{Statistical pretest}

Via the pretest one estimates the fidelity between the true $\rho_{\rm true}$ and the best permutationally invariant state $F_{\rm PI}(\rho_{\rm true})=\max_{\rho_{\rm PI} \geq 0} \tr(\sqrt{\sqrt{\rho_{\rm true}} \rho_{\rm PI}\sqrt{\rho_{\rm true}} })$ using only measurement results from a few settings $\hat a \in T$, \eg, employing only $T=\{\hat e_1,\hat e_2,\hat e_3\}$. Depending on this quantity one decides whether permutationally invariant tomography is worth to continue. As explained in detail in Ref.~\cite{toth10a} this fidelity can be bounded by
\begin{equation}
\label{eq:pretest}
F_{\rm PI}(\rho_{\rm true}) \geq [\tr(\rho_{\rm true} Z)]^2 
\end{equation}
with an operator $Z=\sum z_k^a M_k^a$ being built up by the performed measurements $M_k^a$ given by Eq.~(\ref{eq:PI-measurements}) and satisfying $Z \leq P_{\rm sym}$, where $P_{\rm sym}$ denotes the projector onto the symmetric subspace. 

The expansion coefficients $z_k^a$ should be optimized to attain the best lower bound. For a given target state $\rho_{\rm tar}$ this problem can be cast into a semidefinite program~\cite{cobook,toth09a} that can be solved efficiently using standard numerical routines. However for larger qubit numbers one must again employ the block structure of the measurement operators as given by Eq.~(\ref{eq:PI-meas}) to handle the operator inequality. Note that the projector on the symmetric subspace has a Block structure $P_{{\rm sym},j}=\delta_{j,N/2} \mathbbm{1}$. Then the final problem reads as
\begin{eqnarray}
\label{eq:sdp_pretest}
\max_z && \sum_{a \in T,k} z_k^a\tr(\rho_{\rm tar} M_k^a) \\ 
\nonumber
\textrm{s.t.}&& \sum_{a \in T,k} z_k^a M_{k,j=N/2}^a \leq \mathbbm{1}, \sum_{a \in T,k} z_k^a M_{k,j}^a \leq 0, \forall j < N/2.
\end{eqnarray}

If one experimentally implements this pretest one must account for additional statistical fluctuations. For the chosen $Z$ one can employ the sample mean $\bar Z = \sum z_k^a f_k^a$ using the observed frequencies $f_k^a=n_k^a/N_{\rm R}$ in $N_{\rm R}$ repetitions of setting $\hat a$, as an estimate of the true expectation value $\tr(\rho_{\rm true} Z)$. This sample mean $\bar Z$ will fluctuate around the true mean but large deviations will become less likely, such that for an appropriately chosen $\epsilon$ the quantity $\mathrm{sign}(\bar Z -\epsilon) (\bar Z - \epsilon)^2$ is a lower bound to the true fidelity at the desired confidence level. The proof essentially uses the techniques employed in Refs.~\cite{flammia11a,silva11a}.  

\begin{proposition}[Statistical pretest]
For any $Z=\sum z_k^a M_k^a \leq P_{\rm sym}$ let $\bar Z = \sum z_k^a n_k^a /N_{\rm R}$ denote the sample mean using $N_{\rm R}$ repetitions for setting $\hat a \in T$. If the data are generated by the state $\rho_{\rm true}$ then
\begin{equation}
\mathrm{Prob} \left[ F_{\rm PI}(\rho_{\rm true}) \geq \mathrm{sign}(\bar Z -\epsilon) (\bar Z - \epsilon) ^2 \right] \geq 1-\exp(-2 N_{\rm R} \epsilon^2 / C_z^2)
\end{equation} 
with $C_z^2 = \sum_a \left( z_{\rm max}^s - z_{\rm min}^s \right)^2$ where $z_{\rm max/min}^a$ are the respective optima for setting $\hat a$ over all outcomes $k$.
\end{proposition}

\begin{proof}
The given statement follows along  
\begin{eqnarray}
\mathrm{Prob} &\left[ F_{\rm PI}(\rho_{\rm true}) \geq \mathrm{sign}(\bar Z -\epsilon) (\bar Z - \epsilon) ^2 \right] 
\geq \mathrm{Prob} \left[ \tr(\rho_{\rm true} Z) \geq \bar Z -\epsilon \right] \\ 
&\geq 1 - \mathrm{Prob} \left[ \tr(\rho_{\rm true} Z) \leq \bar Z -\epsilon \right] \geq 1 - \exp(-2 N_{\rm R} \epsilon^2 / C_z^2).
\end{eqnarray}
Here the first inequality holds because the set of outcomes satisfying $\{ n_k^a: \tr(\rho_{\rm true} Z) \geq \bar Z -\epsilon\}$ is a subset of $\{ n_k^a: F_{\rm PI}(\rho_{\rm true}) \geq \mathrm{sign}(\bar Z -\epsilon) (\bar Z - \epsilon) ^2\}$ using Eq.~(\ref{eq:pretest}). In the last inequality we use Hoeffdings tail inequality~\cite{hoeffding} to bound $\mathrm{Prob} \left[ \tr(\rho_{\rm true} Z) \leq \bar Z -\epsilon \right]$. 
\end{proof}

Note that this pretest can also be applied after the whole tomography scheme in which case the projector $P_{\rm sym}=\sum z_k^a M_k^a$ becomes accessible. Moreover, let us point out that a strong statistical significance, or a low $\epsilon$ respectively, might not be achieved with the best expectation value as given by Eq.~(\ref{eq:sdp_pretest})~\cite{jungnitsch10a}; hence optimizing $Z$ for a rather mixed state is often better. 

Finally let us remark that the pretest can be improved by additional projectors, see supplementary material of Ref.~\cite{toth10a}. This leads to the bound $F_{\rm PI}(\rho_{\rm true}) \geq \sum_j p_j^2$ with $p_j$ being the weight of the corresponding spin-$j$ state of the permutationally invariant part of $\rho_{\rm true}$ as given in Eq.~(\ref{eq:PI-state1}). From this expression one sees that this test only works well for states having a rather large weight on one of these spin states. Others, like the totally mixed state, while clearly being permutationally invariant, are not identified as states close to the permutationally invariant subspace. This is in stark contrast to compressed sensing where the certificate succeeds for the whole class of low rank target states~\cite{gross11a}. 

\subsubsection{Entanglement and MaxLik-MaxEnt principle}

Following the last comment from the previous section, we want to argue that even in the case of an inefficient certificate, permutationally invariant state reconstruction as given by Eq.~(\ref{eq:PI_staterecon}) is meaningful. At first we would like to emphasize that the permutationally invariant part of any density operator represents a fair representative to investigate the entanglement properties of the true, unknown state. This is because the transformation given Eq.~(\ref{eq:PI-state}) can be achieved by local operations and classical communications, whereby the amount of entanglement cannot increase. Thus if the permutationally invariant part of the density operator is entangled this holds true also for the real state. 

Second, permutationally invariant state reconstruction also represents the solution of the maximum-likelihood maximum-entropy principle as introduced in Ref.~\cite{teo11a}, which goes as follows: If the performed measurements are not tomographically complete then there is, in general, not a single state $\hat \rho_{\rm ml}$ that maximizes the likelihood but rather a complete set of them. In order to single out a ``good'' representative, Ref.~\cite{teo11a} proposes to choose the state which has the largest entropy, which, according to the Jaynes principle~\cite{jaynes}, is the special state for which one has the fewest information. 

\begin{proposition}[Permutationally invariant MaxLik-MaxEnt principle]\label{prop:PI&MlikMent}
Using the described permutationally invariant tomography scheme, the reconstructed permutationally invariant state given by Eq.~(\ref{eq:PI_staterecon}) (with the likelihood function) is also the solution of the maximum-likelihood maximum-entropy principle. 
\end{proposition}
\begin{proof}
Since the measurements given by Eq.~(\ref{eq:PI-measurements}) are invariant and tomographically complete for permutationally invariant states, all density operators with the same spin ensemble as $\hat \rho_{\rm PI}$ have the same maximum likelihood. According to Ref.~\cite{jaynes} the state with maximal entropy and consistent with a given set of expectation values for operators $M^a_k$ has the form $\rho \propto \exp(\sum_{a,k} \lambda^a_k M^a_k)$. The Lagrange multipliers $\lambda^a_k \in \mathbbm{R}$ must be chosen such that the given expectation values match. However, because all $M^a_k$ are permutationally invariant we can employ the block decomposition given by Eq.~(\ref{eq:PI-meas}) and finally obtain $\exp(\sum_{a,k} \lambda^a_k M^a_k) = \exp(\oplus_j \sum_{a,k} \lambda^a_k M^a_{k,j}\otimes \mathbbm{1}) = \oplus_j \exp(\sum^a_k \lambda^a_k M^a_{k,j}) \otimes \mathbbm{1}$. Hence we obtain the same structure as $\hat \rho_{\rm PI}$, which therefore is also the state with maximum entropy. 
\end{proof}

\section{Conclusion and outline}\label{sec:conclusion}

In this manuscript we provided all necessary ingredients to carry out permutationally invariant tomography~\cite{toth10a} in experiments with large qubit numbers. This includes, besides scheme specific tasks like the statistical pretest and the optimization of the measurement settings, in particular the state reconstruction part.~Accounting for statistical fluctuations due to a finite amount of data, this reconstruction demands the solution of a non-linear large-scale optimization problem. We achieve this by first using a convenient toolbox to store, characterize and process permutationally invariant states, which largely reduces the dimension of the underlying problem and second by using convex optimization, which is superior compared to commonly used numerical routines in many respects. This makes permutationally invariant tomography a complete tomography method requiring only moderate measurement and data analysis effort.

There are many questions one may pursue in this direction: First, let us stress that the current prototype implementation is still not optimal. As explained, the bottleneck is the computation of the second derivatives, hence we strongly believe that Hessian free-optimization, like quasi-Newton or conjugate gradients \cite{geiger99a}, or the use of other barrier functions more tailored to linear matrix inequalities~\cite{pennon} are likely to push the reconstruction limit further. Second, it is natural to try to exploit other symmetries in the development of ``symmetry'' tomography protocols, \ie, tomography should work for all states that remain invariant under the action of a specific group. Clearly any symmetry decreases the number of state dependent parameters, but the challenge is to devise efficient local measurement strategies. Interesting classes here are graph-diagonal~\cite{hein05a} or, more general, locally maximally entangleable states~\cite{kruszynska09a}, translation or shift-invariant states~\cite{oconnor01a}, or $U^{\otimes N}$ invariant states~\cite{eggeling01a,cabello03a}. Third, it is worth to investigate to what extent particularly designed state tomography protocols are useful in further tasks, like process tomography for quantum gates. For instance, permutationally invariant tomography might be unable to resolve the Toffoli gate~\cite{monz09a} directly, but since the operation on all $N$ target qubits is symmetric, a permutationally invariant resolution of this subspace (and the additional control qubit) might be sufficient. Finally let us point out that permutationally invariant tomography can be further restricted to the symmetric subspace, which often contains the desired states. This is reasonable since we have seen that the pretest is only good if the unknown state has a large weight in one of the spin states. However since the symmetric subspace grows only linearly with the number of particles, this tomography scheme can analyse even much more qubits efficiently.

\ackn 
We thank M. Kleinmann, B.~Kraus, T.~Monz, P.~Schindler and J.~Wehr for stimulating discussions about the topic and technicalities. This work has been supported by the FWF (START prize Y376-N16), the EU (Project Q-ESSENCE and Marie Curie CIG 293993/ENFOQI), the BMBF (Chist-Era Project QUASAR), the ERC (Starting Grant GEDENTQOPT), the QCCC of the Elite Network of Bavaria, the Spanish MICINN (Project No. FIS2009-12773-C02-02), the Basque Government (Project No. IT4720-10), and the support of the National Research Fund of Hungary OTKA (Contract No. K83858).

\section*{References}

\bibliographystyle{iopart-num}

\end{document}